\documentclass[runningheads,11pt]{article}
\usepackage{algorithm}
\usepackage{algpseudocode}
\usepackage{amsmath, amsfonts, amsthm}
\usepackage[english]{babel}
\usepackage{booktabs}
\usepackage{bbm}
\usepackage{color}
\usepackage{dsfont}
\usepackage{enumitem}
\usepackage{fancyhdr}
\usepackage{float}
\usepackage[T1]{fontenc}    
\usepackage[colorlinks=true,linkcolor=blue,citecolor=blue]{hyperref}
\usepackage[capitalize,noabbrev]{cleveref}
\usepackage[utf8]{inputenx}
\usepackage[left=1in,top=1in,right=1in,bottom=1in]{geometry}
\usepackage{graphicx}
\usepackage{mathrsfs}
\usepackage{mathtools}
\usepackage{multirow}
\usepackage{nicefrac}       
\usepackage{listings}

\usepackage{sidecap}
\usepackage{subcaption}
\usepackage{textcomp}
\usepackage[textwidth=30mm]{todonotes}
\usepackage{url}       
\usepackage{ulem}
\usepackage{longtable}

\newcommand{\debug}[1]{{\color{black}#1}}

\allowdisplaybreaks

\DeclareMathOperator{\Expect}{\mathbb{E}}
\newcommand{\esp}[1]{\Expect\left[#1\right]}
\newcommand{\prob}[1]{\mathbb P\left(#1\right)}

\newcommand{\dm}{\debug{\mathrm{DM}}}

\newcommand{\va}{\mathrm{VA}}
\newcommand{\vadm}{\debug{\va}}

\newcommand{\crA}{\alpha} 
\newcommand{\crB}{\beta} 
\newcommand{\crC}{\gamma}
\newcommand{\classdp}{\mathrm{ON}}
\newcommand{\opt}{\mathrm{OPT}}
\newcommand{\fk}{\opt}

\newcommand{\pida}{\mathrm{PIDUA}}
\newcommand{\cpi}{\mathrm{PI}}

\newcommand{\Inst}{\mathcal I}

\newcommand{\ber}{\mathrm{Ber}}

\newcommand{\md}{d}

\newtheorem{definition}{Definition}
\newtheorem{theorem}{Theorem}
\newtheorem{lemma}{Lemma}
\newtheorem{remark}{Remark}
\newtheorem{corollary}{Corollary}
\newtheorem{proposition}{Proposition}

\newtheorem*{thmstar}{\thmlabel}
\newcommand{\thmlabel}{}
\NewDocumentEnvironment{thmfree}{m o}
  {\renewcommand{\thmlabel}{#1}%
   \IfNoValueTF{#2}
     {\begin{thmstar}}
     {\begin{thmstar}[#2]}}
  {\end{thmstar}}

\usepackage[T1]{fontenc}
\usepackage{authblk}

\title{\textbf{Prophet Inequalities with Delayed and Uncertain Acceptance}}

\author[1]{\textbf{Emile Martinez}}
\author[1]{\textbf{Felipe Garrido-Lucero}}
\author[1]{\textbf{Umberto Grandi}}
\author[2]{\textbf{Sebastián Pérez-Salazar}}

\affil[1]{IRIT, Université Toulouse Capitole, Toulouse, France}
\affil[2]{Rice University, Houston TX, USA}

\date{}

\begin{document}
	%
	
	\maketitle              
	%
	\begin{abstract}
		We introduce the \textit{prophet inequality with delayed and uncertain acceptance} model, a variant of the classical prophet inequality in which a decision-maker sequentially evaluates options whose acceptance is uncertain and revealed only after a random delay. At each time step, the decision-maker observes the realized value of the arriving option and must irrevocably decide whether to select it or continue searching. Once an option is selected, the process is suspended for a delay drawn from a known distribution with expectation $\md$, during which no other option can be considered. When the delay expires, the selection succeeds with a known probability. If it succeeds, the decision-maker receives the realized value and the process terminates; otherwise, the search resumes.
		
		In addition to the online decision-maker, we consider two stronger benchmarks: the value-aware decision-maker, who knows all value realizations in advance but not the acceptance outcomes, and the prophet, who knows both the values and the acceptance outcomes. We show that, in the worst case, both decision-makers achieve a $1/(\md + 2)$ fraction of the prophet's value, while the decision-maker achieves half of the value-aware decision-maker's value.
		
		We then turn to the no-delay case, where we establish an equivalence between our model and the classical prophet inequality, and show that the value-aware decision-maker beats the $1/2$ barrier against the prophet whenever all acceptance probabilities are bounded away from $0$.
		\medskip
		
		\noindent\small{\textbf{Keywords}. Prophet Inequalities $\cdot$ Uncertain Acceptance $\cdot$ Random Delay.}
	\end{abstract}
	%
	%
	%
	\setcounter{page}{1}
	
	\section{Introduction}
	
	Consider a student searching for a flat in a highly competitive rental market. Flats are visited sequentially, and hesitation is a luxury she cannot afford: after each visit, the student must make an immediate and irreversible decision about whether to apply. Although advertisements may provide prior information about an apartment’s quality, its true value is revealed only upon inspection. Passing on an apartment means losing it forever, while accepting it terminates the search. The student thus faces a fundamental tension induced by the online nature of the problem, namely the trade-off between accepting a currently observed option and forgoing it in the hope of encountering a better apartment in the future. The online selection literature has devoted substantial effort to understanding and alleviating this tension across a wide range of applications, including classical online selection problems \cite{dynkin1963optimum}, $k$-selection \cite{correa2021optimal}, online knapsack \cite{chakrabarty2008online}, and online matching \cite{mehta2010online}.
	
	One of the simplest and most influential models capturing this tension is the \textit{prophet inequality $\mathrm{(}\cpi\mathrm{)}$ problem} introduced by Krengel and Sucheston~\cite{krengel1977semiamarts}. In a $\cpi$ problem, a decision-maker (the student) observes a sequence of independent, nonnegative random variables $X_1,\ldots,X_n$, representing apartment qualities, each drawn from a known distribution. At each step $i$, upon observing $X_i$, the decision-maker must irrevocably choose whether to accept the value---thereby terminating the process with payoff $X_i$---or to reject it and continue, without the possibility of returning to previously rejected options. The goal in this type of problems is to design an algorithm that maximizes the expected payoff. The performance of such an algorithm is measured by its \textit{competitive ratio}, defined as the worst-case ratio between the algorithm’s expected payoff 
	and the expected maximum value achievable in hindsight $\esp{\max_{i \in [n]} X_i}$, 
	often referred to as the value of a \textit{prophet}. A classical result shows that there exists an algorithm achieving a competitive ratio of at least $1/2$, and that this guarantee is tight (see e.g.~\cite{kleinberg2012matroid,krengel1977semiamarts,samuel1984comparison}).
	
	While the online trade-off faced by the decision-maker in the classical setting, where acceptance is immediate and certain, is well understood, competitive rental markets introduce new trade-offs for the decision-maker. In particular, submitting an application does not guarantee immediate or certain acceptance: landlords often process multiple applications simultaneously, so applicants typically face delays before receiving either an acceptance or rejection decision (see, e.g.,~\cite{regnaud2025rental}). In our introductory example, an application may be rejected with positive probability, forcing the student (the decision-maker in our model) to continue searching. This gives rise to two new trade-offs. First, an apartment of modest quality but with a high probability of acceptance may be preferable to a more attractive apartment that is unlikely to accept the student's application. Second, because acceptance decisions are delayed, applying to one apartment temporarily prevents the decision-maker from applying to others, which may disappear from the market while waiting for a response. Motivated by these observations, we introduce the \textit{prophet inequalities with delayed and uncertain acceptance} ($\pida$) problem, which extends the $\cpi$ framework to account for both the possibility of rejection and the delayed response to applications. At time $i$, after observing $X_i$, the decision-maker may either continue searching or select the current option. If the latter is chosen, the selection succeeds with a given probability after a random delay $D_i \in \mathbb{N}$ of known distribution; otherwise, it is rejected, and the search resumes at time $i+D_i+1$. During this waiting period, no other options can be considered, naturally capturing situations in which a student may decide not to invest additional effort in visiting apartments after submitting an application, or where applications are binding, preventing the student from withdrawing an accepted offer. A formal description of the model is provided in \Cref{sec:model}.
	
	Our model is not restricted to apartment search. In academic publishing, a manuscript submitted to a journal faces an uncertain acceptance outcome after an unknown review period, during which the authors must wait for the editorial decision before submitting elsewhere. Similarly,  
	in labor markets, newly hired workers typically undergo a probationary period, possibly extended, during which they usually refrain from pursuing alternative employment opportunities.
	
	Interestingly, due to the extra source of uncertainty, $\pida$ induces an agent between the decision-maker and the prophet, that is, an agent that is aware of the values but is blind to acceptance: the \textit{value-aware-decision-maker} ($\vadm$), which can be interpreted in our leading example as a real estate expert with a perfect valuation of flats but which remains uncertain of the contingency of being rejected. The existence of this intermediate agent raises additional research questions: \textit{How well does the $\vadm$ perform in general? How much stronger is the $\vadm$ than the decision-maker? How weak is the $\vadm$ relative to the prophet? Is there an actual separation between the decision-maker and the $\vadm$ or knowing the values in advance actually does not help?} In order to answer these questions, we will study the three natural competitive ratios arising from the comparison of three agents by establishing worst-case guarantees.
	
	\subsection{Contributions and Techniques} 
	
	The main result of this article is the characterization of the three worst-case competitive ratios under i.i.d. delays with expected value $\md$.
	
	\begin{thmfree}{Worst-case competitive ratios}[\Cref{thm:cr}]
		The worst-case competitive ratios of both decision-makers against the prophet are ${1}/{(\md+2)}$. In the worst case, the decision-maker achieves $1/2$ of the value of the $\vadm$.
	\end{thmfree}
	
	The proof of \Cref{thm:cr} splits into two cases: the comparison of both decision-makers with the prophet yielding ${1}/{(\md+2)}$ competitive-ratios, and the comparison between the decision-maker and the $\vadm$ yielding a ${1}/{2}$ competitive-ratio. The former ones are established by using the elementary renewal theorem, allowing to compute the number of random variables that both decision-makers must miss due to the cumulated delays in order to achieve a $1/(\md+2)$-fraction of the prophet reward. Interestingly, once fixed the value $\md$ of the expected delay, the tightness in \Cref{thm:cr} can be proved for \textbf{any delay distribution} whose expected value equals $\md$. The competitive ratio between the two decision-makers comes from observing that the value of the $\vadm$ corresponds to the sum of two possible reward sources, while the decision-maker can only opt for the best of the two.
	
	\Cref{thm:cr} shows the opportunity costs that our two decision-makers incur due to the existence of delayed and uncertain acceptance, which behaves linearly in the expected delay $\md$ with respect to the $1/2$-guarantee in classical $\cpi$. In particular, under the no-delay case (that is, $\md$ $=$ $0$ as delays are non-negative), all bounds collapse to $1/2$. Still, notice that $\pida$ instances retain uncertain acceptance even under zero delay. This suggests that uncertain acceptance alone does not make the problem harder than $\cpi$, further confirmed by our \textit{reduction lemma}.
	
	
	
	
	\begin{thmfree}{Reduction Lemma}[\Cref{lemma:reduction_lemma}]
		Under no delay ($\md$ $=$ $0$), for any $\pida$ instance, knowing in advance whether an application will be accepted at the moment of choosing provides no additional advantage to the decision-maker. In particular, every $\pida$ instance without delay can be reduced to a $\cpi$ instance whose observed random variables are given by the product of the realized values and their acceptance indicators.
	\end{thmfree}
	
	The reduction lemma reveals a fundamental distinction that arises from the delay. In $\pida$ without delay, a rejected application incurs no opportunity cost: once the decision-maker observes that an application has failed, it can immediately resume the search. In particular, $\pida$ without delay reduces to $\cpi$, so classical results for $\cpi$ transfer directly to our setting, e.g., the worst-case optimality of single-threshold algorithms (\Cref{prop:thresholds_are_optimal}). By contrast, under positive delay this reduction no longer holds, since submitting an application may block the decision-maker for several rounds, regardless of whether it is eventually accepted. This observation motivates the development of new techniques to establish the worst-case competitive ratios.
	
	The worst-case competitive ratios of \Cref{thm:cr}, in the no-delay case, highlight an intriguing phenomenon. In the worst case, the decision-maker guarantees half the utility of the $\vadm$, who in turn guarantees half the utility of the prophet. Yet, the decision-maker also guarantees half the utility of the prophet directly. This naturally raises the question of whether knowing the realized values in advance provides any worst-case advantage. As we show, the answer is subtle. Even when acceptance probabilities are arbitrarily small, the $\vadm$ may still perform arbitrarily close to the prophet. Nevertheless, much stronger insights can be obtained when acceptance probabilities are bounded away from zero, as illustrated by the following result.
	
	\begin{thmfree}{$\vadm$ versus prophet under no delay}(\Cref{cor:AB_on_bounded_acceptance})
		Given a fixed $\pida$ instance, let $p$ be the minimum acceptance probability. Then, the $\vadm$ can always guarantee at least ${1}/({2-p})$ of the value of the prophet in the instance.
	\end{thmfree}
	
	To prove \cref{cor:AB_on_bounded_acceptance}, we establish a larger result for the classical $\cpi$ case over scaled Bernoulli random variables (\Cref{thm:CR_bernoulli}), which we prove by building a sequence of instance modifications allowing to reduce any given instance to an optimization problem whose optimal value characterizes the analyzed worst-case competitive ratio.

	\subsection{Related Literature}

	\noindent\textbf{Prophet inequalities.} $\cpi$ were introduced by Krengel and Sucheston \cite{krengel1977semiamarts}, who proved that for any instance, a decision-maker taking optimal decisions can always guarantee at least half of the utility obtained by the prophet. Samuel-Cahn \cite{samuel1984comparison} later showed that a competitive ratio of $1/2$ can be achieved by a single-threshold algorithm, and that this guarantee is optimal. More recently, $\cpi$ have attracted considerable attention in mechanism design due to their close connection with posted-price mechanisms, where reservation prices naturally translate into threshold-based policies (see, e.g.,~\cite{chawla2010multi,correa2019pricing,hajiaghayi2007automated}). The work most closely related to ours is that of Assaf et al.~\cite{assaf1998statistical}, who studied a $\cpi$ model in which an online decision-maker observes a random variable $X_i$ but an additional random variable $A_i$ remains hidden. Upon selecting an option, the decision-maker receives a payoff $f(X_i,A_i)$ for a given function $f$. Our model, in the extreme regime of having deterministic delays equal to $n-1$, corresponds to the particular case when $A_i$ is binary and $f(x,a) = x\cdot a$. Like us, Assaf et al.~\cite{assaf1998statistical} considered an analogous benchmark to our $\vadm$, who observes all value realizations but not the hidden variables $A_i$. They show that the decision-maker achieves a competitive ratio of $1/2$ against the $\vadm$, but cannot obtain any constant competitive ratio against the prophet, both bounds being captured by our article.
	\smallskip
	
	\noindent\textbf{Uncertain acceptance.} There is a large body of work exploring uncertain offer acceptance in hiring problems, starting with the secretary problem \cite{smith1975secretary,tamaki1991secretary}. More recently, Perez-Salazar et al.~\cite{perez2025robust} consider a competitive framework similar to ours; however, a key difference is that they assume adversarially chosen values revealed in \textit{random order}, whereas we only assume distributional information. Moreover, in this line of work the probability of acceptance of all random variables is assumed to be \textit{homogeneous}. In contrast, our model does not impose such homogeneity, allowing us to capture richer settings (e.g., high-value distributions may be less likely to be accepted). Recently, Xu et al.~\cite{xu2026online} studied an online selection problem in which acceptance does not end the process and the reward is the sum of the values of the selected items up to having an unsuccessful selection, in which case the game is over.
	\smallskip
	
	\noindent 
	\textbf{Delayed acceptance.} In online settings, delay is often a resource for the decision-maker, who may wait to observe future arrivals before committing (see, e.g.,~\cite{emek2016online}). Closer to our work, the \textit{PI over time} considers a decision-maker who commits to a selected value for several rounds~\cite{abels2025prophet,disser2020hiring,perez2024optimal}, and reusable-resource models block the resource for some time after each allocation~\cite{faw2022learning}. Similarly to us, their delay corresponds to an opportunity cost. However, unlike our setting, the selected value is always collected. Moreover, our delay is random and follows a pre-fixed distribution.
	\smallskip
	
	\noindent\textbf{Competition in online selection problems}. It is natural to interpret uncertain acceptance as the outcome of competition, where an item may be claimed by another agent before the decision-maker successfully secures it. Competitive variants of online selection problems have received increasing attention, particularly in the secretary problem \cite{immorlica2006secretary,karlin2015competitive,ramsey2024stackelberg}. Closer to our work, a number of papers have investigated competitive versions of $\cpi$. Ezra et al.~\cite{ezra2021prophet} studied a setting in which $k$ agents simultaneously compete over the same $\cpi$ instance. When several agents attempt to select the same item, a tie-breaking rule determines the winner, while the unsuccessful agents continue the process. Along similar lines, Gensbittel et al.~\cite{gensbittel2024competition} focused on the case of two competing agents, allowing each agent to select any previously observed value provided it has not already been claimed by the opponent. From the perspective of an individual agent, these competitive models naturally induce a $\cpi$ with uncertain (and potentially delayed) acceptance. However, our setting is fundamentally different, as we assume independence among acceptances, unlike \cite{gensbittel2024competition} where rejection decreases the number of competitors and therefore, increases future probabilities of acceptance.
	\medskip
	
	
	
	
	\section{Model}\label{sec:model}
	
	In this section, we introduce our model, the three agents in our setting, and the three corresponding competitive ratios. Throughout the article, we will denote $\mathbb{N}$ the set of non-negative integer numbers, $[k] := \{1,...,k\}$ for any $k \in \mathbb{N}$, $\ber(p)$ a Bernoulli random variable with parameter $p\in [0,1]$, and $a^+ := \max(a, 0)$, for any $a \in \mathbb{R}$.
	
	
	
	
	\begin{definition}
		A \textbf{prophet inequality with delayed and uncertain acceptance} $(\pida)$ instance consists of a sequence of $n$ independent random vectors $\Inst:= (Y_1,\ldots, Y_n)$, where for any $i \in [n]$, $Y_i := (X_i, A_i, D_i)$, with $X_i$ being a non-negative random variable of finite expectation, named the value, $A_i = \ber(p_i)$ with $p_i \in [0,1]$, named the acceptance, and $D_i$ a random variable taking values on $\mathbb{N}$, named the delay, and such that all random variables are mutually independent. Moreover, we assume that all delays $(D_i)_{i\in [n]}$ are identically distributed and have finite expected value $\md$.
	\end{definition}
	
	
	In this setting, we consider three classes of algorithms based on the information known from the instance $\Inst$, each of them going through the sequence $Y_1,\ldots,Y_n$ sequentially.
	\medskip
	
	\noindent\textbf{Decision-Maker}. A (fully) online algorithm, named \textit{decision-maker} ($\dm$), which, at time $i$, observes the realization of $X_i$ and must decide whether to select it or to continue to the next time step. If $X_i$ is selected, first $D_i$ is revealed, and then $A_i$ is revealed at time $i + D_i + 1$. The process ends whenever $A_i = 1$, with the algorithm receiving $X_i$ as payoff; otherwise, the process resumes at time $i+D_i+1$ (unless $i+D_i+1 > n$, in which case the game ends and the decision-maker obtains $0$). Under no selection, the process resumes at time $i+1$.
	
	\begin{definition}\label{def:dynamic_programming}
		For any $\pida$ instance $\Inst = (Y_1,...,Y_n)$, we denote $\dm(\Inst)$ the \textbf{optimal expected value} of the $\dm$. That is, the expected value obtained by the decision-maker when using dynamic programming in the instance $\Inst$. Writing $V_i := \dm(Y_i, \dots, Y_n)$ the expected value of the $\dm$ at time $i \in [n]$, and extending $V_i=0$ for $i > n$, it holds,
		\begin{equation}\label{eq:dynamic_programming}
			\hspace{-0.4cm}V_i = \mathbb{E}\bigl[{\max\bigl(V_{i+1},\ p_i\cdot X_i + (1 - p_i)\cdot\Expect[V_{i+D_i+1}]\bigr)}\bigr],
		\end{equation}
		where $\Expect[V_{i+D_i+1}]$ is taken over the randomness of $D_i$.
	\end{definition}
	
	
	\noindent\textbf{Value-Aware Decision-Maker}. A value-aware online algorithm, named \textit{value-aware decision-maker} ($\vadm$), that knows the realizations ($x_1,...,x_n)$ of the values $(X_1,\ldots,X_n)$ upfront but does not know the acceptances $(A_1,\ldots,A_n)$ nor the delays $(D_1,...,D_n)$. At time $i$, without observing the realization of $A_i$, the $\vadm$ can try to select $X_i$, and the process continues/stop as in the previous class. Note that, since the $\vadm$ knows the realizations of the random variables $(X_1,...,X_n)$, its optimal expected value is given by applying dynamic programming on $((x_1,A_1,D_1),...,(x_n,A_n,D_n))$. 
	
	\begin{definition}\label{def:VA_DM_value}
		Let $\va(\Inst)$ be the optimal expected value of the $\vadm$ on the instance $\Inst$. Given $x_i$ the realization of $X_i$, for any $i \in [n]$, and denoting 
		\begin{equation*}
			\va\text{-}\dm_{\Inst}(x_1, \dots, x_n) := \dm((x_1, A_1,D_1), \dots, (x_n, A_n, D_n)),
		\end{equation*}   
		it follows $\va(\Inst) := \Expect[\va\text{-}\dm_{\Inst}(X_1,\linebreak[1] \dots, X_n)]$.
	\end{definition}
	
	
	\noindent\textbf{Prophet}. A full-knowledge offline algorithm that knows the realizations of all random vectors $Y_1,\ldots,Y_n$ upfront, named the \textit{prophet}. 
	We denote $\fk(\Inst)$ the expected value of the prophet on instance $\Inst$, formally given by 
	$$\fk(\Inst) := \mathbb{E}[\max_{i\in[n]} A_i\cdot X_i].$$
	
	Note that the prophet, since it knows the acceptance realizations, is not affected by the delay. The three considered classes of algorithms, and their respective best algorithms, give rise to the following three different \textit{competitive ratios}. 
	
	\begin{definition}
		Given $\Inst$ a $\pida$ instance, we define
		$$
		\crA(\Inst) := \frac{\dm(\Inst)}{\va(\Inst)}, \crB(\Inst) :=  \frac{\va(\Inst)}{\fk(\Inst)}, \text{ and } \crC(\Inst) := \frac{\dm(\Inst)}{\fk(\Inst)}.$$
	\end{definition}
	Note that for any instance $\Inst$, $\crA(\Inst),\crB(\Inst),\crC(\Inst) \in [0,1]$, and it always holds that $\crC(\Inst) \leq \crA(\Inst)$ and $\crC(\Inst) \leq \crB(\Inst)$.
	\medskip
	
	\noindent\textbf{Notation for the classical setting.} We denote by $\classdp$ the expected utility achieved by the best online algorithm, i.e., dynamic programming, in the classical prophet inequality ($\cpi$) framework. Regarding the prophet, we abuse notation and simply write $\opt$ for its expected utility.
	
	\section{Worst-Case Competitive Ratios}\label{sec:worst_case_competitive_ratios}
	
	This section is devoted to studying the worst-case competitive ratios. The following result is one of the main contributions of the article.
	
	\begin{theorem}\label{thm:cr}
		For any $\pida$ instance $\Inst$, it holds 
		$$\crB(\Inst) \geq \frac{1}{\md+2} \text{ and } \crC(\Inst) \geq \frac{1}{\md+2}, \text{ while } \crA(\Inst) \geq \frac{1}{2}.$$ 
		Moreover, all inequalities are tight, and this holds for every delay distribution with expected value $d$.
	\end{theorem}
	\Cref{thm:cr} shows the opportunity costs that our two decision-makers incur due to the existence of the delay. Interestingly, the error made behaves \textit{linearly} on $\md$, the expected delay, with respect to the $1$/$2$-guarantee on classical $\cpi$. In particular, we can identify two extreme cases:
	\vspace{-0.2cm}
	\begin{itemize}[leftmargin = *,itemsep = -1mm]
		\item[$\bullet$] When the delay is deterministic and equal to $n-1$, we recover the result established by Assaf et al.~\cite{assaf1998statistical} that whenever the $\dm$ can only partially observe the reward realization and has no option to continue in case of rejection, her worst-case competitive ratio against the prophet is non-constant and decreases as $1/n$. 
		\item[$\bullet$] When the delay is deterministic and equal to $0$, our bound recovers the classical result of Krengel and Sucheston~\cite{krengel1977semiamarts}. Still, notice that $\pida$ instances retain uncertain acceptance even under zero delay. This suggests that uncertain acceptance alone does not make the problem harder than the classical $\cpi$. 
		We return to this discussion in \Cref{sec:no_delay}, where we formally study the no delay setting.
	\end{itemize}
	
	The rest of this section proves \Cref{thm:cr}. We split the analysis into $(\crB,\crC)$ (\Cref{sec:DMs_versus_Prophet}) and $\crA$ (\Cref{sec:DM_versus_VA}).
	
	\subsection{Decision-makers versus Prophet}\label{sec:DMs_versus_Prophet}
	
	We begin by proving the lower bounds stated in \Cref{thm:cr} for $\crB$ and $\crC$. Since for any instance $\Inst$, $\crC(\Inst) \leq \crB(\Inst)$, it is sufficient to prove the following.
	
	\begin{proposition}\label{prop:gamma_lower_bound}
		For any $\pida$ instance $\Inst$, $\crC(\Inst) \geq {1}/{(\md+2)}.$
	\end{proposition}
	
	\begin{proof}
		Let $\Inst$ be an instance. Denote $Z_i:=\Expect[V_{i+D_i+1}]$. For $i\in [n]$, it follows,
		\begin{equation}\label{eq:V_i_as_function_of_Z_i}
			V_i = \mathbb{E}[{V_{i+1} + (p_iX_i + (1-p_i)Z_i - V_{i+1})^+}].
		\end{equation}
		For $i\in [n]$, notice that,
		$$\setlength\arraycolsep{2pt}\begin{array}{rcl} 
			p_i(X_i - V_1)^+ &\leq & (p_i(X_i-Z_i))^+ \\
			& = &  \big(p_iX_i + (1-p_i)Z_i - V_{i+1} + V_{i+1} - Z_i\big)^+ \\
			& \leq & (p_iX_i + (1-p_i)Z_i - V_{i+1})^+ + V_{i+1} - Z_i,
		\end{array}$$
		where the first inequality comes from $V_1 \geq Z_i$ and the last one from $V_{i+1}-Z_i \geq 0$. Then, taking the expectation and plugging \Cref{eq:V_i_as_function_of_Z_i}, we get,
		\begin{equation}
			\label{eq:minoration_proof_d+2}
			p_i\esp{(X_i-V_1)^+} \leq V_i - Z_i, \quad \text{ for } i\in[n].
		\end{equation}
		Define next $\Delta_i := V_i - V_{i+1}\geq 0$, for $i\in[n]$. Notice that
		$$\setlength\arraycolsep{2pt}\begin{array}{rcl}
			V_i - Z_i  & = & \Expect[V_i - V_{i+D_i+1}] = \Expect\left[\sum\limits_{j=i}^{i+D_i} \Delta_j\right] = \sum\limits_{j=i}^{n} \Delta_j \cdot \prob{D_i \geq j-i},
		\end{array}$$
		where we have used that $\Delta_j = 0$ for $j > n$. Therefore 
		\begin{align*}
				\sum\limits_{i=1}^n (V_i-Z_i) = \sum\limits_{j=1}^n\Delta_j\sum\limits_{i = 1}^j \prob{D_i \geq j-i} \leq (\md + 1)\sum\limits_{j=1}^n\Delta_j  =  (\md+1)V_1,
		\end{align*}
		where the inequality comes from recalling that delays are i.i.d and,
		\begin{equation}\label{eq:IID_delay}
			\sum\limits_{i = 1}^j \prob{D_i \geq j-i} \leq \sum\limits_{r=0}^{\infty} \prob{D_1 \geq r} = \md + 1.    
		\end{equation}
		Finally, for any $i\in [n]$, we have 
		\begin{align*} 
			\hspace{-0.5cm}A_iX_i & \leq V_1 + (A_iX_i - V_1)^+ 
			\leq V_1 + \sum\limits_{j=1}^n A_j(X_j-V_1)^+.
		\end{align*}
		Thus, 
		$$\max_{i\in[n]} A_iX_i \leq V_1 + \sum\limits_{j=1}^n A_j(X_j-V_1)^+.$$
		Therefore, taking expectation and using \Cref{eq:minoration_proof_d+2}, we get,
		$$\opt(\Inst) \leq V_1 + \sum\limits_{i=1}^n(V_i-Z_i) \leq  (\md + 2)V_1, $$
		which proves the result as $V_1 = \dm(\Inst)$.
	\end{proof}
	
	\begin{remark}
		The proof of \Cref{prop:gamma_lower_bound} uses that the delays are identically distributed only in \Cref{eq:IID_delay}. A stochastic dominance argument extends this step to independent but non-identically distributed delays $(D_1,\dots,D_n)$. Let $\bar{D}$ be the random variable defined by
		$$\mathbb{P}(\bar{D} \geq r) = \max\limits_{i \in [n]} \mathbb{P}(D_i \geq r), \text{ for any } r \geq 0,$$
		which is well defined as the right-hand side is a non-increasing tail function equal to $1$ at $r = 0$. Then $\bar{D}$ stochastically dominates each $D_i$, and every instance $\Inst$ satisfies $$\crC(\Inst) \geq \frac{1}{(\mathbb{E}[\bar{D}] + 2)}.$$ However, extending the remaining results of the article, in particular the tightness of this bound, to non-identically distributed delays is far from straightforward, so we restrict our attention to i.i.d. delays.
	\end{remark}
	
	\begin{proposition}\label{prop:tightness_gamma}
		There exists a sequence of $\pida$ instances $(\Inst_\eta)_\eta$, with arbitrary delay distributions, such that $\crB(\Inst_\eta), \crC(\Inst_\eta) \to 1/(\md+2)$.
	\end{proposition}
	
	\begin{proof}
		Given $n \in \mathbb N$, with $n \geq 1$, and $\varepsilon \in (0,1)$, define the instance $\Inst(n, \varepsilon) = (Y_1,...,Y_{n+1})$, given by 
		\begin{equation*}
			Y_1 := \bigl(\frac{n}{\md+1}, \ber(1), D_1\bigr) \text{ and } Y_i := \bigl(\frac{1}{\varepsilon}, \ber(\varepsilon), D_i\bigr),
		\end{equation*}
		for any $i$ $\in$ $\{2,...,n+1\}$, that is, all random variables $X_i$ are deterministic, $A_1$ has acceptance probability $1$ while the rest $A_i$ have acceptance probability $\varepsilon$, and delays follow an arbitrary distribution of expected value $d$. 
		Denote $V_i(n, \varepsilon)$ the value of the decision-maker when faced to the $i$-th random vector in the instance $\Inst(n, \varepsilon)$. For $i \geq 2$, it holds,
		$$V_i(n, \varepsilon) = 1 + (1-\varepsilon)\Expect\bigl[V_{i+D_i+1}(n, \varepsilon)\bigr].$$
		Denote next, for $\ell \in \mathbb Z$,
		\begin{equation*}
			N_\ell := \Expect\bigl[\min \bigl\{ k \in \mathbb{N}, \text{ such that } \sum\limits_{j=1}^k D_j + 1 \geq \ell \bigr\}\bigr],
	\end{equation*}
	to be the expected number of times the decision-maker needs to fail in order to miss $\ell$ random variables due to the delays. By backward induction on $i$,
	(refer to Appendix \ref{sec:missing_proof} for a formal proof)
	it follows,
	$$\lim\limits_{\varepsilon \to 0} V_{i}(n, \varepsilon) = N_{n+2-i}.$$
	In particular, 
	\begin{equation*}
		\lim_{\varepsilon\to 0} V_1(n,\varepsilon) =  \lim_{\varepsilon\to 0} \max \bigl\{\frac{n}{\md+1}, V_2(n,\varepsilon)\bigr\} = \max \bigl\{\frac{n}{\md+1}, N_n\bigr\}.
	\end{equation*}
	On the other hand, 
	\begin{equation*}
		\opt(\Inst(n, \varepsilon)) =  (1-(1-\varepsilon)^n)\cdot\frac{1}{\varepsilon} + (1-\varepsilon)^n\cdot\frac{n}{\md+1},
	\end{equation*}
	therefore, $\lim_{\varepsilon \to 0} \opt(\Inst(n, \varepsilon)) = n + \frac{n}{\md+1}$. The proof is concluded by noticing that
	\begin{equation*}
		\lim_{\varepsilon\to 0} \crC(\Inst(n,\varepsilon)) = \max\biggl(\frac{1}{\md+1}, \frac{N_n}{n}\biggr)\cdot\frac{\md+1}{\md+2},
	\end{equation*}
	and therefore, $\lim_{n\to \infty}  \lim_{\varepsilon\to 0} \crC(\Inst(n,\varepsilon)) = {1}/{(\md+2)}$,
	where we have used that $N_n/n$ tends to $1/(\md+1)$ when $n$ grows from the elementary renewal theorem \cite{doob1948renewal}. Finally, since the variables $(X_1,...,X_n)$ in the considered instance are deterministic, we conclude the tightness of $\crB$ as well.
\end{proof}



\subsection{$\dm$ versus $\vadm$.}\label{sec:DM_versus_VA}

This section is devoted to proving that the remaining worst-case competitive ratio is equal to $1/2$. Notice that, since classical $\cpi$ is a particular case of our model, where the $\vadm$ coincides with the prophet, the upper bound for $\crA$ is obtained for free. To prove the lower bound, we require to prove first a technical result. Given a $\pida$ instance
\begin{equation*}
	\Inst = (Y_1,...,Y_n) = ((X_1,A_1,D_1),...,(X_n,A_n,D_n)),
\end{equation*}
define, for each $i \in [n]$, a \textbf{selection threshold}
$$\tau_i := \frac{1}{p_i}\bigl[V_{i+1} - (1-p_i)\cdot\mathbb{E}[V_{i+D_i+1}]\bigr],$$
where we assume w.l.o.g. that $p_i > 0$. Finally, denote $m_i := \esp{(X_i-\tau_i)^+}$.

\begin{lemma}\label{lem:replace_shifted_longshot}
	Given a $\pida$ instance $\Inst$ and $q \in (0,1]$, define, for an arbitrary $i \in [n]$, $Y_i^{(q)}:= (X_i^{(q)},A_i,D_i)$, where,
	$$X_i^{(q)} := \tau_i + L_i^{(q)}, \text{ with } L_i^{(q)}:=\frac{m_i}{q}\ber(q).$$
	Then, the $\pida$ instance $\Inst^{(q)} := (Y_1, ..., Y_{i-1}, Y_i^{(q)}, Y_{i+1},..., Y_n)$
	has the same selection thresholds as $\Inst$, and,
	$$\lim\limits_{q\to0} \crA(\Inst^{(q)}) \leq \crA(\Inst).$$
\end{lemma}

\begin{proof}
	Define $V_j$ and $V_j^{(q)}$ the future expected values of the $\dm$ when faced to the $j$-th random vector, respectively, in $\Inst$ and $\Inst^{(q)}$. We prove, first of all, that for any $j \in [n]$, both values are equivalent. Note that this is indeed the case for any $j \in \{i+1,...,n\}$. Next, by backward induction,
	$$\setlength\arraycolsep{2pt}\begin{array}{rcl}
		V_i^{(q)} & = & \esp{\max\left(V_{i+1}, p_i X_i^{(q)} + (1-p_i) \mathbb{E}\bigl[{V_{i+D_i+1}}\bigr]\right)}\\[0.5em]
		& = & \esp{ p_i X_i^{(q)} + (1-p_i) \mathbb{E}\bigl[{V_{i+D_i+1}}\bigr]}\\[0.5em]
		& = & V_{i+1} + p_i \esp{(X_i - \tau_i)^+}\\[0.5em]
		& = & V_{i+1} + \mathbb{E}[(p_iX_i + (1-p_i)\mathbb{E}[V_{i+D_i+1}] - V_{i+1})^+]\\[0.5em]
		& = &  \mathbb{E}[\max(V_{i+1}, p_i X_i + (1-p_i) \mathbb{E}[{V_{i+D_i+1}}]) \\[0.5em]
		& = & V_i,
	\end{array}$$
	where the first of the equalities comes from applying the induction hypothesis and the second one since $X_i^{(q)} \geq \tau_i$ almost surely. Consequently, induction implies that $V_j = V_j^{(q)}$ for any $j < i$. Since all values coincide between $\Inst$ and $\Inst^{(q)}$, the thresholds remain unchanged, and $\dm(\Inst) = \dm(\Inst^{(q)})$. With all this in mind, to conclude the proof it is sufficient to show that $\lim_{q\to 0} \va(\Inst^{(q)}) \geq \va(\Inst)$.
	Let us denote $X_{-i} := (X_j)_{j\in[n]\setminus \{i\}}$ and define,
	$$f_{X_{-i}}(X_i) := \va\text{-}\dm_\Inst(X_1, \dots, X_n),$$
	with $\va\text{-}\dm$ as in \Cref{def:VA_DM_value}. Notice that, since $\Inst$ and $\Inst^{(q)}$ share the same acceptances and delays, it follows,
	\begin{equation*}
		\va(\Inst) = \mathbb{E}[f_{X_{-i}}(X_i)] \text{ and } \va(\Inst^{(q)}) = \mathbb{E}[f_{X_{-i}}(X_i^{(q)})].
	\end{equation*}
	Let us fix a realization $x_{-i}$ of $X_{-i}$. We aim at proving
	\begin{equation*}
		\mathbb{E}[f_{x_{-i}}(X_i)] \leq \lim\limits_{q\to 0} \mathbb{E}[f_{x_{-i}}(X_i^{(q)})],
	\end{equation*}
	since then the result can be concluded from the law of total expectation and the dominated convergence theorem. Let $\pi$ be a strategy of the $\vadm$. Since she knows the value realizations, it follows that $\pi$ corresponds to a subset of indices from $[n]$ representing the random vectors to accept if the $\vadm$ reaches them. For a given $i \in \pi$, denote $q_i^\pi$ the probability that the $\vadm$ reaches $i$ under the policy $\pi$, that is,
	\begin{equation*}
		q_i^\pi = \sum_{h \in H} \prod_{\ell = 1}^{|h|-1} \prob{(j_\ell+D_{j_\ell}, j_{\ell+1}] \cap \pi = \{j_{\ell+1}\}} (1-p_{j_\ell}),
	\end{equation*}
	where $H$ is the set of all paths $h = \{j_1,...,j_{|h|}\}$ of selected items leading to $Y_i$ under policy $\pi$. It follows, 
	\begin{equation*}
		\va\text{-}\dm_\Inst(x_1,...,x_n) = \max_{\pi\subseteq [n]} \sum\limits_{i\in\pi} q_i^\pi\cdot p_i\cdot x_i.
	\end{equation*}
	In particular, denoting
	$$a_\pi := q_i^\pi p_i \text{ and } b_\pi := \sum\limits_{j\in \pi\setminus \{i\}} q_j^\pi p_jx_j,$$
	we obtain, $f_{x_{-i}}(x) = \max_{\pi\subseteq [n]} \{a_\pi x + b_\pi\}$
	and, therefore,
	$$\lim_{x\to\infty} \frac{1}{x}f_{x_{-i}}(x) = \max_{\pi\subseteq [n]} a_\pi.$$
	With this in mind, it follows,
	\begin{equation*}
		\esp{f_{x_{-i}}\left(X_i^{(q)}\right)} = (1-q) f_{x_{-i}}(\tau_i) + q f_{x_{-i}}\left(\tau_i+\frac{m_i}{q}\right), 
	\end{equation*}
	so $\lim_{q\to 0} \mathbb{E}[f_{x_{-i}}(X_i^{(q)})] =  f_{x_{-i}}(\tau_i) + (\max_{\pi\subseteq [n]} a_\pi)\cdot m_i$. We conclude the proof by observing that 
	$$f_{x_{-i}}(x) \leq f_{x_{-i}}(\tau_i) + (\max\limits_{\pi\subseteq [n]} a_\pi) \cdot (x-\tau_i)^+,$$
	as the inequality holds for $x \leq \tau_i$ since $f_{x_{-i}}$ is increasing, and for $x \geq \tau_i$ since $f_{x_{-i}}$ is a maximum of affine functions.
\end{proof}

\noindent We are ready to prove the lower bound for $\crA$.

\begin{proposition}\label{prop:lower_bound_between_DMs}
	For any $\pida$ instance $\Inst$, $\crA(\Inst) \geq 1/2.$
\end{proposition}

\begin{proof}
	Let $\Inst$ be a $\pida$ instance, $\varepsilon > 0$ and $q \in (0,1]$. Applying sequentially \Cref{lem:replace_shifted_longshot}, as the selection thresholds remain unchanged, there exists $q_1, \dots, q_n \in (0,q]$, such that for $i\in[n]$,
	$
	\crA(Y_1^{(q_1)},...,Y_i^{(q_i)}, Y_{i+1},...,Y_n)\leq \crA(\Inst) + i \cdot \frac{\varepsilon}{n},
	$
	where $Y_i^{(q_i)} := (X_i^{(q_i)},A_i,D_i)$. Note that,
	\begin{align*}
		\va(Y_1^{(q_1)},...,Y_n^{(q_n)}) &\leq \va((\tau_1,A_1,D_1),...,(\tau_n,A_n,D_n))  + \va((L_1^{(q_1)}\hspace{-0.05cm},A_1,D_1),...,(L_n^{(q_n)}\hspace{-0.05cm},A_n,D_n))\\
		&\leq \dm((\tau_1, A_1,D_1), \dots, (\tau_n, A_n,D_n)) + \sum\limits_{i=1}^n p_im_i,
	\end{align*}
	where the last inequality comes from the fact that $\vadm$ in an instance with deterministic values becomes identical to the decision-maker. 
	Similarly, it follows,
	\begin{equation*}
		\begin{split}
			\dm(Y_1^{(q_1)},...,Y_n^{(q_n)}) &\geq \max\bigl\{\dm((\tau_1,A_1,D_1),...,(\tau_n,A_n,D_n)),\dm((L_1^{(q_1)}\hspace{-0.1cm},A_1,D_1),...,(L_n^{(q_n)}\hspace{-0.1cm},A_n,D_n)) \bigr\}\\
			&\geq \max\bigl\{ \dm((\tau_1,A_1,D_1), ..., (\tau_n, A_n,D_n)) , (1-q)^n\sum\limits_{i=1}^n p_i m_i\bigr\}\\
			&\geq (1-q)^n\max\bigl\{ \dm((\tau_1,A_1,D_1), ..., (\tau_n, A_n,D_n)) , \sum\limits_{i=1}^n p_i m_i\bigr\}\\
			&\geq \frac{(1-q)^n}{2}\bigl(\dm((\tau_1,A_1,D_1), ..., (\tau_n, A_n,D_n)) + \sum\limits_{i=1}^n p_i m_i\bigr)
		\end{split}
	\end{equation*}
	where the second inequality comes from lower-bounding the value of $\dm$ by the strategy that only selects the first non-zero variable. 
	Formally, 
	\begin{equation*}
		\begin{split}
			\dm( (L_1^{(q_1)}\hspace{-0.2cm}, A_1,D_1), \dots, (L_n^{(q_n)}\hspace{-0.2cm}, A_n,D_n)) &\geq \sum\limits_{i=1}^n \prod\limits_{j=1}^{i-1}(1-q_j) q_i\frac{m_i}{q_i}p_i \\
			&\geq \sum\limits_{i=1}^{n} (1-q)^{i-1}p_i m_i \geq (1-q)^n\sum\limits_{i=1}^n p_i m_i.       
		\end{split}
	\end{equation*}
	With all this in mind, it follows,
	$$\crA(\Inst) \geq \crA\bigl(Y_1^{(q_1)},...,Y_n^{(q_n)}\bigr) - \varepsilon \geq \frac{(1-q)^n}{2} - \varepsilon,$$
	and we conclude taking both $\varepsilon$ and $q$ to $0$.
\end{proof}

The ratio between the $\dm$ and the $\vadm$ equals $1/2$ for every delay distribution, recovering the classical guarantee of~\cite{krengel1977semiamarts}. Both agents face the same acceptance uncertainty and incur the same opportunity cost from delays, so the gap between them reduces to the classical value of knowing the values in advance. This suggests that uncertain acceptance alone does not make $\pida$ harder than $\cpi$, an intuition we make precise in \Cref{sec:no_delay}.


\section{Uncertain Acceptance without Delay}\label{sec:no_delay}

This section is devoted to studying $\pida$ instances without delay. We begin by formalizing one of the insights obtained from \Cref{thm:cr}, namely, that instances with uncertain acceptance but no delay do not involve any extra difficulty with respect to classical $\cpi$. 




\begin{lemma}[Reduction Lemma]\label{lemma:reduction_lemma}
	Let $\Inst = (Y_1,...,Y_n)$ be a $\pida$ instance where all delays $(D_1,...,D_n)$ are deterministic and equal to $0$. Denote, for $i \in [n]$, $Z_i := X_i \cdot A_i$, and set $\Inst' := (Z_1,...,Z_n)$. It holds that $\dm(\Inst) = \classdp(\Inst')$, where $\classdp$ corresponds to the optimal online algorithm in the classical $\cpi$ setting.
\end{lemma}
\begin{proof}
	Note that, almost surely, for $i\in[n]$, $(A_iX_i-A_iV_{i+1})^+ = (A_iX_i - V_{i+1})^+$ as $A_i \in\{0,1\}$. Therefore, as $$\setlength\arraycolsep{2pt}\begin{array}{rl}
		V_i & = V_{i+1} + \esp{(A_iX_i + (1-A_i)V_{i+1} - V_{i+1})^+} \\
		& = V_{i+1} + \esp{(A_iX_i  -A_i V_{i+1})^+},
	\end{array}$$ 
	we can conclude by backward induction.
\end{proof}

The reduction lemma comes from the fact that the decision-maker incurs no error when selecting a value and being rejected as she can still try to get the following random variable. In particular, there is no difference on knowing the realization of $A_i$ at time $i$ prior to selecting it. 
The equivalence established by the reduction lemma allows to bring results from the $\cpi$ framework to $\pida$ under no delay, such as the following one.

\begin{corollary}\label{prop:thresholds_are_optimal}
	Single-threshold algorithms are worst-case optimal for $\pida$ instances without delay.
\end{corollary}

A single-threshold algorithm with threshold $\tau$ accepts the first random variable whose realization exceeds it. Suppose that $x_i \geq \tau$, for some $i \in [n]$. If $A_i = 1$, then the $\dm$ receives exactly the same payoff in both our model and the classical setting. On the other hand, if $A_i = 0$, the decision-maker continues searching in both models: in our setting because the application is rejected, and in the classical setting because the corresponding realization corresponded to $z_i = 0$, and therefore the condition $z_i \geq \tau$ is not satisfied.


\subsection{Positioning the $\vadm$}
\label{sec:vadm_position}

Under no delay, that is, $\md = 0$ since delays are non-negative, \Cref{thm:cr} yields a seemingly paradoxical result. On the one hand, all three worst-case competitive ratios equal $1/2$. On the other hand, every instance $\Inst$ satisfies $\crA(\Inst)\cdot\crB(\Inst) = \crC(\Inst)$. Hence, the worst-case instances of the three competitive ratios must differ, since a single instance attaining $1/2$ for both $\crA$ and $\crB$ would force $\crC(\Inst) = 1/4$. This observation raises the question of characterizing the instances on which knowing the realized values in advance actually helps the $\vadm$ close the gap to the prophet, as opposed to those on which this additional information provides no advantage, making the $\vadm$ effectively as weak as the $\dm$. A first intuition might be that, when acceptance probabilities are high, the $\vadm$ has essentially as much information as the prophet, whereas, when they are low, the $\vadm$ loses most of its informational advantage and behaves similarly to the $\dm$. As we show, only the first half of this intuition is correct: high acceptance probabilities do bring the $\vadm$ close to the prophet (\Cref{cor:AB_on_bounded_acceptance}) while low acceptance probabilities do not make it behave necessarily like the $\dm$ (\Cref{prop:end_of_hopes}).

\begin{proposition}\label{prop:end_of_hopes}
	There exist $\pida$ instances $\Inst = (Y_1,...,Y_n)$ without delay and arbitrary acceptance probabilities verifying $\crA(\Inst) \to {1}/{2}$ and $\crB(\Inst) \to 1$.
\end{proposition}

The proof of \Cref{prop:end_of_hopes} is included in Appendix \ref{sec:missing_proof}. To address the case of high acceptance, we characterize a sufficient condition for the $\vadm$ to beat the $1/2$ barrier as a function of the lowest acceptance probability. 

\begin{theorem}\label{cor:AB_on_bounded_acceptance}
	Let $p \in [0,1]$ be fixed and $\Inst = (Y_1,...,Y_n)$ an arbitrary $\pida$ instance without delay, such that, for any $i \in [n]$, $p_i \geq p$. Then, $\crB(\Inst) \geq {1}/{(2-p)}$. Moreover, this bound is tight.
\end{theorem}

\Cref{cor:AB_on_bounded_acceptance} is proved by establishing a stronger result on the $\cpi$ setting with \textit{scaled Bernoulli random variables}.

\begin{theorem}\label{thm:CR_bernoulli}
	Consider a $\cpi$ instance $\Inst = (X_1,...,X_n)$ such that, for any $i \in [n]$, $X_i = \lambda_i\cdot \ber(p_i)$, where $\lambda_i \in \mathbb{R}_+$ and $p_i \in (0,1]$ are fixed scalars. Then, it holds,
	$$\frac{\classdp(\Inst)}{\opt(\Inst)} \geq \frac{1}{2 - \min_{i\in [n]}p_i}.$$
	Moreover, this bound is tight and can be achieved using a single-threshold algorithm.
\end{theorem}

\Cref{thm:CR_bernoulli} is the main contribution of this section which directly implies \Cref{cor:AB_on_bounded_acceptance} when observing that the $\vadm$ in a $\pida$ instance without delay and the decision-maker in a $\cpi$ instance with scaled Bernoulli random variables coincide.
To prove \Cref{thm:CR_bernoulli}, we employ a similar approach to Ekbatani et al.~\cite{ekbatani2024prophet}, which studies $\cpi$ with the possibility of buying back previously selected values. Through a sequence of reductions, they show that it suffices to analyze instances with scaled Bernoulli random variables, enabling the problem to be solved via a linear program. We use a sequence of technical results, that we present in the following list:
\vspace{-0.18cm}
\begin{itemize}[leftmargin = *, itemsep = -1mm]
	\item[$\bullet$] \Cref{lem:ordered setting} (originally proved in~\cite{ekbatani2024prophet}) proves that, in the context of \Cref{thm:CR_bernoulli}, the worst instances for $\classdp$ are those where the values $(\lambda_1,...,\lambda_n)$ are increasingly sorted.
	\item[$\bullet$] \Cref{lem:ordered -> threshold} proves that, under increasingly sorted values $(\lambda_1,...,\lambda_n)$, single-threshold algorithms are optimal.
	\item[$\bullet$] \Cref{lem:const then inst} proves that in increasingly sorted instances, there exists $i \in [n]$ such that we can replace the first $i-1$ random variables by a deterministic one, obtaining a worst instance for $\classdp$. Moreover, in such instance, $\classdp$ is forced to pick the first value.
	\item[$\bullet$] \Cref{lem:optim problem} leverages the previous results to embed the problem into a simple optimization problem whose solution corresponds to a $\cpi$ instance as in \Cref{thm:CR_bernoulli}, where $\classdp$ obtains value $1$ and the value of the prophet is maximized.
	\item[$\bullet$] \Cref{lem:relaxation B_n,lem:first is one} show that in the previous construction, the last random variable can be replaced by a random variable of mean $1$.
	\item[$\bullet$] \Cref{lem:1 for all i} shows the instance can be reduced to only having two random variables, one deterministically equal to $1$ and one scaled Bernoulli with success probability $p$, where $\classdp$ obtains value $1$ and the prophet gets $2-p$.
\end{itemize}

\noindent The technical lemmas and the proof of \Cref{thm:CR_bernoulli} are included in Appendix \ref{sec:proof_CR_bernoulli}.

\section{Conclusions}

We introduce the \textit{prophet inequality with delayed and uncertain acceptance} model, an extension of prophet inequality in which selecting an option may fail after its value has been observed. Moreover, the acceptance outcome is revealed only after a random delay, during which the decision-maker cannot observe any other random variable. This setting naturally gives rise to a new benchmark lying between the decision-maker and the prophet: the \textit{value-aware decision-maker} ($\vadm$), who knows all value realizations in advance but not the acceptance outcomes. We derive tight worst-case competitive ratios between each pair of agents, thereby generalizing both the classical ${1}/{2}$ guarantee of Krengel and Sucheston~\cite{krengel1977semiamarts} and the ${1}/{n}$ guarantee of Assaf et al.~\cite{assaf1998statistical}. Furthermore, in the no-delay setting, we identify sufficient conditions under which the $\vadm$ surpasses the ${1}/{2}$ barrier, a first step towards characterizing the instances on which advance knowledge of the realized values provides a genuine advantage over the $\dm$.
\medskip

\noindent\textbf{Future work}. Our results open several research directions. First, our analysis assumes full independence between all random variables. On the one hand, we expect our setting to extend to correlation within each random vector $Y_i$, with all results remaining valid. On the other hand, correlation between the acceptances $(A_1,\dots,A_n)$ leads to a drastic drop in the competitive ratio (refer to  \Cref{prop:corr_Ai} in Appendix \ref{sec:independence}).
Second, it remains open whether the guarantees improve in the case of i.i.d. rewards. A slight adaptation of the tight instance of \Cref{thm:cr} increases the upper bound from ${1}/{(\md+2)}$ to ${1}/{(\md+1)}$. However, this bound cannot be tight, since it does not recover the known guarantee of $0.745$ in the no-delay case \cite{correa2017posted}. Other classical questions from the prophet inequality literature remain open in our model, notably the worst-case optimality of single-threshold rules in the general delay case. Finally, necessary conditions under which the $\vadm$ surpasses the ${1}/{2}$ barrier in the no-delay setting are still to be established, and the generalization of our results to the case of non identically distributed delays. 
\bibliographystyle{splncs04}
\small{
	\bibliography{biblio}}

%
	
	
	
	

\normalsize
\appendix

\section{Correlated Acceptance}\label{sec:independence}

Throughout the main article, all random variables of an instance $\Inst$ are assumed to be independent. While we believe that allowing correlation within each random vector $Y_i$ leaves our results unchanged, we show here that correlation across different random vectors drives all worst-case competitive ratios to zero. For correlated values $(X_1,\dots,X_n)$, this was already established by Hill and Kertz~\cite{hill1981ratio}. The following result shows that the same holds for correlated acceptances $(A_1,\dots,A_n)$.

\begin{proposition}\label{prop:corr_Ai}
	There exists a $\pida$ instance $\Inst$ with correlated acceptance and no delay such that $\crC(\Inst) = O(1/n)$.
\end{proposition}

\begin{proof}
	Let $p \in (0,1)$ be fixed. Define the instance $\Inst = (Y_1,...,Y_n)$ where each $X_i$ is deterministic and equal to $p^{-i+1}$, and $(A_1,...,A_n)$ verifies, for any $j \in [n]$,
	\begin{equation*}
		\begin{split}
			&\prob{(A_1, \dots, A_j) = 1 \text{ and } (A_{j+1},...,A_n) = 0 } = p^{j-1} - p^{j},\\
			&\prob{(A_1, \dots, A_n) = (1,1, \dots, 1)} = p^{n-1}.
		\end{split}
	\end{equation*}
	As proved by Hill and Kertz~\cite{hill1981ratio}, the competitive ratio between the decision-maker and the prophet is of order $\mathcal{O}(\frac{1}{n(1-p)})$, in particular vanishing as $n\to \infty$.
\end{proof}






\section{Missing Proofs}\label{sec:missing_proof}

\noindent\textbf{\Cref{prop:tightness_gamma}}. \textit{In the context of the proof of \Cref{prop:tightness_gamma}, it holds $$\lim\limits_{\varepsilon \to 0} V_{i}(n, \varepsilon) = N_{n+2-i}.$$}

\begin{proof}
	We prove the result by backward induction. Let $i\geq 2$. By induction hypothesis, for $j > i$, $\lim_{\varepsilon \to 0} V_{j}(n, \varepsilon) = N_{n+2-j}$. Then, 
	\begin{align*}
		\esp{N_{n+1-i-D_i}} & = \sum_{k \in \mathbb N}\prob{D_{i}=k} N_{n+1-i-k}\\
		& = \sum_{k = 0}^{n-i} \prob{D_i=k} N_{n+1-i-k}\\
		& = \sum\limits_{k=0}^{n-i} \prob{D_i=k} \lim_{\varepsilon \to 0} V_{i+k+1}(n, \varepsilon)\\
		& = \lim_{\varepsilon \to 0}  \sum\limits_{k=0}^{n-i} \prob{D_i=k} V_{i+k+1}(n, \varepsilon)\\
		& = \lim_{\varepsilon \to 0}  \sum\limits_{k\in \mathbb N} \prob{D_i=k} V_{i+k+1}(n, \varepsilon)\\
		& = \lim_{\varepsilon \to 0} \esp{V_{i+D_i+1}(n, \varepsilon)}
	\end{align*}
	where the second equality holds as $N_l = 0$ if $l\leq 0$, the third by induction hypothesis, the fourth as the limit of a finite sum and the fifth as $V_j = 0$ for $j > n+1$. Therefore, $$\lim_{\varepsilon \to 0} V_{i}(n, \varepsilon) = 1 + \esp{N_{n+1-i-D_i}}.$$
	\noindent As all delays are i.i.d., we can define for $\ell\in \mathbb Z$ \begin{equation*}
		N_\ell = \Expect\biggl[\min \bigl\{ k \text{ such that } \sum\limits_{j=2}^{k+1} D_j + 1 > \ell \bigr\}\biggr]
	\end{equation*}
	and therefore,
	\begin{align*}
		\esp{N_{n+1-i-D_i}} 
		&= \esp{N_{n+1-i-D_1}}\\
		&= \esp{\Expect\biggl[\min \bigl\{ k \text{ such that } \!\!\sum\limits_{j=2}^{k+1} D_j \!+\! 1 > n\!+\! 1\!-\!i\!-\!D_1 \bigr\} \bigg| D_1 \biggr] }\\
		&= \Expect\biggl[\min \bigl\{ k \text{ such that } \sum\limits_{j=2}^{k+1} D_j \!+\! 1 > n\!+\! 1\!-\!i\!-\!D_1 \bigr\}  \biggr]\\
		&= \Expect\biggl[\min \bigl\{ k \text{ such that } \sum\limits_{j=1}^{k+1} D_j \!+\! 1 > n\!+\! 1\!-\!i \!+\!1\bigr\}  \biggr]\\
		&= \Expect\biggl[\min \bigl\{ k \text{ such that } \sum\limits_{j=1}^{k} D_j \!+\! 1 > n\!+\!2\!-\!i \bigr\}  \biggr] - 1\\
		&= N_{n+2-i} - 1
	\end{align*}
	where the third equality is by the law of total expectation.
\end{proof}



\noindent\textbf{\Cref{prop:end_of_hopes}}. \textit{There exist $\pida$ instances $\Inst = (Y_1,...,Y_n)$ without delay and arbitrary acceptance probabilities verifying $\crA(\Inst) \to {1}/{2}$ and $\crB(\Inst) \to 1$.}

\begin{proof}
	Fix $\varepsilon > 0$ and $p,p_L \in (0,1)$. Consider the instance $\Inst:= (Y_1,...,Y_n, Y_L)$ such that, for $i \in [n]$, 
	\begin{equation*}
		Y_i = (1, \ber(p),0), \text{ and } Y_L = \bigl(\frac{1}{\varepsilon p_L}\ber(\varepsilon), \ber(p_L),0\bigr),
	\end{equation*}
	that is, all $X_i$ are deterministic and equal to $1$, all $A_i$ have acceptance probability $p$, $X_L$ is a long-shot, $A_L$ has acceptance probability $p_L$, and all delays are deterministic and equal to $0$. It follows that $\dm(\Inst) = 1$, $\va(\Inst) = 1 + (1-\varepsilon)(1-(1-p)^n)$, and $\fk(\Inst) = 1 + (1-\varepsilon p_L)(1-(1-p)^n)$.
	In particular, taking $n\to\infty$ and $\varepsilon \to 0$, we get, $\crA(\Inst) \to {1}/{2}$ and $\crB(\Inst) \to 1$.
\end{proof}

\subsection{Proof of \Cref{thm:CR_bernoulli}}\label{sec:proof_CR_bernoulli}

In the following, let $\lambda_1,...,\lambda_n \in \mathbb{R}_+$, $p_1,...,p_n \in [0,1]$, $X_i = \lambda_i \cdot \ber(p_i)$ to be a scaled Bernoulli random variable, for $i \in [n]$, and $\Inst = (X_1,...,X_n)$ the $\cpi$ instance defined by these random variables.

\begin{lemma}\label{lem:ordered setting}
	Let $\sigma \in \Sigma([n])$ be a permutation of $[n]$ such that $\lambda_{\sigma(1)} \leq \dots \leq \lambda_{\sigma(n)}$. Denote $\Inst_\sigma := (X_{\sigma(1)}, \dots, X_{\sigma(n)})$, that is, the instance where the random variables are increasingly sorted by their values $\lambda_i$. It follows,
	$$
	\frac{\classdp(\Inst)}{\opt(\Inst)} \geq \frac{\classdp(\Inst_\sigma)}{\opt(\Inst_\sigma)}.
	$$
\end{lemma}

Ekbatani et al.~\cite{ekbatani2024prophet} showed a similar result in their setting to Lemma \ref{lem:ordered setting}. For the sake of completeness, we include the proof of our result.

\begin{proof}
	Consider $\Inst$ as stated, and suppose that $\lambda_i > \lambda_{i+1}$ for some $i < n$. We will prove that by exchanging the order of $X_i$ and $X_{i+1}$, the value of $\classdp$ decreases. We proceed by cases.
	\smallskip
	
	\noindent\textbf{Case 1.} Suppose $\lambda_{i+1} \leq \classdp(X_{i+2}, \dots, X_n)$. It follows,
	\begin{align*}
		\classdp(X_{i}, X_{i+1}, \, \dots, \,X_n) &= \esp{\max\left(X_i, \esp{\max(X_{i+1}, \classdp(X_{i+2}, \dots, X_n))}\right)} \\
		&= \esp{\max(X_i, \classdp(X_{i+2}, \dots, X_n))} \\
		&= \esp{\max\left(X_{i+1}, \esp{\max(X_{i}, \classdp(X_{i+2}, \dots, X_n))}\right)} \\
		&= \classdp(X_{i+1}, X_{i}, X_{i+2}, \dots, X_n).
	\end{align*}
	
	\noindent\textbf{Case 2.} Suppose $\lambda_{i+1} > \classdp(X_{i+2}, \dots, X_n)$. It follows, 
	$$\begin{array}{c}
		\classdp(X_{i}, X_{i+1}, \, \dots, \,X_n)  = p_{i}\lambda_i + (1-p_i)p_{i+1}\lambda_{i+1} + (1-p_i)(1-p_{i+1})\classdp(X_{i+2}, \dots, X_n),
	\end{array}$$
	while 
	\begin{align}\label{eq:swapped_instance}
		\classdp(X_{i+1}, X_{i}, X_{i+2}, ..., X_n) = p_{i+1}\lambda_{i+1} + (1-p_{i+1})p_{i}\lambda_{i} + (1-p_{i+1})(1-p_{i})\classdp(X_{i+2}, ..., X_n).
	\end{align}
	Suppose additionally that $\lambda_{i+1} \leq \classdp(X_i, X_{i+2}, ..., X_n)$. Equation ($\ref{eq:swapped_instance}$) becomes,
	$$
	\classdp(X_{i+1}, X_{i}, X_{i+2}, \dots, X_n) = p_{i}\lambda_i + (1-p_i)\classdp(X_{i+2}, \dots, X_n).
	$$
	In such case, 
	$$
	\classdp(X_{i}, X_{i+1}, \, \dots, \,X_n)  -\classdp(X_{i+1}, X_{i}, X_{i+2}, \dots, X_n) \geq 0,
	$$ 
	as recall $\lambda_{i+1} > \classdp(X_{i+2}, \dots, X_n)$ since we are in case 2. Suppose finally that $$\lambda_{i+1} > \classdp(X_i, X_{i+2}, \dots, X_n),$$ in which case Equation ($\ref{eq:swapped_instance}$) becomes,
	$$\begin{array}{c}
		\classdp(X_{i+1}, X_{i}, X_{i+2}, \dots, X_n) = p_{i+1}\lambda_{i+1}+ (1-p_{i+1})p_{i}\lambda_{i} + (1-p_{i+1})(1-p_{i})\classdp(X_{i+2}, \dots, X_n).
	\end{array}$$
	We conclude then,
	$$\begin{array}{l}
		\classdp(X_{i}, X_{i+1}, \, \dots, \,X_n) - \classdp(X_{i+1}, X_{i}, X_{i+2}, \dots, X_n) = p_i p_{i+1} (\lambda_i - \lambda_{i+1}) \geq 0.
	\end{array}$$
	By monotonicity of the expectation and the maximum, by induction, we deduce that $$\classdp(X_1, \dots, X_n) \geq \classdp(X_1, \dots, X_{i-1}, X_{i+1}, X_i, X_{i+2}, \dots, X_n).$$ Since the value of the prophet is agnostic to the order, we conclude the proof.
\end{proof}

From Lemma \ref{lem:ordered setting} we can suppose, without loss of generality, that the instance $\Inst$ is increasingly sorted. We call such instances \textit{ordered instances}.

\begin{lemma}\label{lem:ordered -> threshold}
	Let $\Inst = (X_1,...,X_n)$, where each $X_i = \lambda_i\cdot\ber(p_i)$, be an ordered instance, that is, $\lambda_1 \leq ... \leq \lambda_n$. There exists a value $\tau \in \mathbb{R}$ such that the single-threshold algorithm with threshold $\tau$ is optimal in $\Inst$.
\end{lemma}

\begin{proof}
	Suppose that $\classdp$ selects the $i$-th random variable. Notice that this holds if and only if $\lambda_i \geq \classdp(X_{i+1},...,X_n)$, as otherwise $\classdp$ has an incentive to continue searching. In particular, since the instance is ordered and by the definition of $\classdp$, it follows,
	\begin{align*}
		\lambda_{i+1} \geq \lambda_i \geq \classdp(X_{i+1}, \dots, X_n) \geq \classdp(X_{i+2}, \dots, X_n).
	\end{align*}
	We observe therefore, that whenever $\classdp$ picks a random variable $X_i$, the algorithm has the incentive to choose any of the subsequent random variables that is not zero, as remember, random variables in $\Inst$ are scaled Bernoulli variables. Consider $\tau := \lambda_i$, the value chosen by $\classdp$. Because of the way threshold algorithms pick a value, and since the instance is ordered, the algorithm presents exactly the same behavior as $\classdp$, that is, reject everything until $i$, and then choose the first scaled Bernoulli variable whose realization is not zero.
\end{proof}

Lemma \ref{lem:ordered -> threshold} states that threshold algorithms, for a well-chosen threshold, are optimal for ordered instances with scaled Bernoulli random variables. The following result shows that, given this \textit{well-chosen} threshold, we can modify the instance in order to further decrease the value obtained by the decision-maker.

\begin{lemma}\label{lem:const then inst}
	Let $\Inst$ be an ordered instance with scaled Bernoulli random variables, and $\tau$ an optimal threshold, that is, the $\tau$-single-threshold algorithm achieves the value $v:= \classdp(\Inst)$. Consider the instance $\Inst_v := (X_v, X_i,...,X_n)$, where $X_v$ is deterministically equal to $v$ and $i \in [n]$ is such that $\lambda_{i-1} < \tau \leq \lambda_i$. It follows,
	\begin{align*}
		\frac{\classdp(\Inst)}{\opt(\Inst)} \geq \frac{\classdp(\Inst_v)}{\opt(\Inst_v)}.
	\end{align*}
\end{lemma}

\begin{proof}
	Since the optimal threshold $\tau$ verifies $\lambda_{i-1} < \tau \leq \lambda_i$, for any $j < i$, it follows $$\classdp(X_j,\dots, X_n) = \classdp(X_{j+1}, \dots, X_n),$$ that is, the optimal thing to do is to reject the first $i-1$ values. By induction, it follows $$\classdp(X_1, \dots, X_n) = \classdp(X_{i}, \dots, X_n).$$ In particular, considering the instance $\Inst_v$ as in the statement, it follows, $\classdp(\Inst) = \classdp(\Inst_v)$. Finally, noticing that 
	\begin{align*}
		\opt(\Inst_v) &= \esp{\max\{v,X_i,X_{i+1},...,X_n\}} \\
		&= \esp{\max\{ v,X_1,...,X_n\}}\\
		&\geq \esp{\max\{X_1,...,X_n\}}\\
		&= \opt(\Inst),
	\end{align*}
	we conclude the proof.
\end{proof}

A classical concern in the design of optimal algorithms for prophet inequalities is to force the algorithm to keep looking for future options and not to take the first random variable that seems good. Lemma \ref{lem:const then inst} creates worst-case instances by attacking exactly this point by putting at the beginning of the instance a random variable with deterministic value equal to the future expected reward of $\classdp$. In particular, since the ratio between the expected values of $\classdp$ and $\opt$ is invariant to multiplication by scalar, without loss of generality, in the following we consider instances $\Inst = (X_1,X_2,...,X_n)$, where $X_1 = 1$ deterministically, and each $X_i = \lambda_i\cdot\ber(p_i)$, for $i \in \{2,...,n\}$, with $1 < \lambda_2 < ... < \lambda_n$. Moreover, we suppose that 
\begin{align}\label{eq:online_opt_choose_1}
	\classdp(X_2,...,X_n) \leq 1,
\end{align}
that is, optimal online algorithms always pick the first random variable. 

Our following result builds on the previous construction to lower bound the competitive ratio of $\classdp$ by the optimal value of an optimization problem over both, the success probabilities $(p_2,...,p_n)$ and the scaling constants $(\lambda_2,...,\lambda_n)$. In order to consider Equation (\ref{eq:online_opt_choose_1}) into the optimization, we introduce $F_i := F(X_i, \dots, X_n)$, for any $i \in [n]$, to be the expected value of the first positive value among $X_i, \dots, X_n$. Notice $F_i$ is a function of $(p_i,...,p_n)$ and $(\lambda_i,...,\lambda_n)$ only. In particular, Equation (\ref{eq:online_opt_choose_1}) holds if and only if $F_i \leq 1$, for any $i \in [n]$ as, by Lemma \ref{lem:ordered -> threshold}, $\classdp(X_2, \dots, X_n) \in \{F_1, \dots, F_n\}$.

\begin{lemma}\label{lem:optim problem}
	Consider the following optimization problem,
	\begin{align*}\label{eq:B_n}
		\begin{array}{cll}
			\max & \sum\limits_{i=1}^n\prod\limits_{j>i}(1-p_j)p_i\lambda_i \\
			s.t. & p_1 = 1, \lambda_1 = 1\\ 
			& p_i \in [p, 1], & 2 \leq i \leq n\\
			& \lambda_i \geq \lambda_{i-1}, & 2 \leq i \leq n\\
			& F_i \leq 1 & 1 \leq i \leq n\\
		\end{array}
		\tag{$\mathcal B_n$}
	\end{align*}    
	Then, for any instance $\Inst$, it always holds,
	\begin{align*}
		\frac{\classdp(\Inst)}{\opt(\Inst)} \geq \frac{1}{\mathrm{val}(\mathcal B_n)},
	\end{align*}
	where $\mathrm{val}(\mathcal{B}_n)$ is the optimal value of the optimization problem ($\mathcal{B}_n$).
\end{lemma}

\begin{proof}
	Note the constraints of problem (\ref{eq:B_n}) define an ordered instance whose first random variable is deterministically equal to $1$ and such that $\classdp$ always picks it. In particular, in such instance, namely $\Inst_{\mathcal{B}_n}$, $\classdp(\Inst_{\mathcal{B}_n}) = 1$. Regarding the prophet, since the instance is ordered, it follows,
	\begin{align*}
		\!\!\!\!\!\!\!\! \esp{\max_{i \in[k]} X_i} = p_k\lambda_k + (1-p_k)\esp{\max_{i \in[k-1]}X_i}, \forall k \in [n],
	\end{align*}
	and therefore, by induction,
	\begin{align*}
		\opt(\Inst_{\mathcal{B}_n}) =  \esp{\max_{i \in[n]} X_i} = \sum\limits_{i=1}^n\prod\limits_{j>i}(1-p_j)p_i\lambda_i.
	\end{align*}
	In particular, we obtain $\opt(\Inst_{\mathcal{B}_n}) \leq \mathrm{val}(\mathcal{B}_n)$. We conclude by applying the previous lemmas, as for any instance $\Inst$, it follows 
	\begin{align*}
		\frac{\classdp(\Inst)}{\opt(\Inst)} \geq \frac{\classdp(\Inst_{\mathcal{B}_n})}{\opt(\Inst_{\mathcal{B}_n})} \geq \frac{1}{\mathrm{val}(\mathcal{B}_n)},
	\end{align*}
	concluding the proof.
\end{proof}

Our last technical results consist in upper-bounding $\mathrm{val}(\mathcal{B}_n)$ by relaxing its constraints and showing that there always exist optimal solutions with a simple structure. First, notice the values $(F_i)_{i\in [n]}$ verify the following recurrence relation:
\begin{align}
	\begin{split}\label{rel:rec F}
		F_n & = p_n\lambda_n \\
		F_i & = p_i\lambda_i + (1-p_i)F_{i+1}, \quad 1 \leq i \leq n-1 \\
		F_i & = \sum_{j=i}^n \prod_{k = i}^{j-1} (1-p_k)p_j\lambda_j, \quad 1 \leq i \leq n
	\end{split}
\end{align}

First, by relaxing some of the constraints in Problem \eqref{eq:B_n}, the following result holds.

\begin{lemma}\label{lem:relaxation B_n}
	Consider the optimization problem     
	\begin{align}
		\tag{$\mathcal B_n'$}
		\label{eq:B_n'}
		\begin{array}{cll}
			\max & \sum\limits_{i=1}^n\prod\limits_{j>i}(1-p_j)p_i\lambda_i \\
			s.t. & p_1 = 1, \lambda_1 = 1 &\\
			& p_i \in [p, 1], & 2 \leq i \leq n\\
			& \lambda_i \geq 1, & 2 \leq i \leq n\\
			& F_i \leq 1, & 1\leq i \leq n 
		\end{array}
	\end{align}    
	It follows that $\mathrm{val}\eqref{eq:B_n} \leq \mathrm{val}\eqref{eq:B_n'}$.    
\end{lemma}

Next, in the relaxed optimization problem, we can always fix the last $F_n$ to be equal to $1$.

\begin{lemma}\label{lem:first is one}
	There exists an optimal solution of problem \eqref{eq:B_n'} such that $F_n = 1$.
\end{lemma}

\begin{proof}
	Let $(\mathbf{p},\boldsymbol{\lambda}) = (p_1,...,p_n,\lambda_1,...,\lambda_n)$ be an optimal solution of problem \eqref{eq:B_n'}, and suppose $F_n < 1$. We claim we can increase $F_n$ keeping all other $F_i$ bounded by $1$, and to arrive to a contradiction by constructing a feasible solution with higher value. In order to do it, from the recurrence relation, we need to increase $\lambda_n$ while letting $\mathbf{p}$ untouched. Consider 
	$$k := \max\{i\in [n] \mid F_i = 1\}.$$
	It follows that $k \geq 1$, since $F_1 = 1$, and that for any $i > k$, $F_i < 1$. Suppose $p_k = 1$. In particular, $$F_k = p_k\lambda_k + (1-p_k)F_{k+1} = p_k\lambda_k.$$ Then, increasing $\lambda_n$ leaves $F_i$, for all $i \leq k$, unchanged, as none of them depend on $\lambda_n$. Since $F_i < 1$ for $i > k$, there exists a gap $\varepsilon$ such that by increasing $\lambda_n$ by $\varepsilon$, the constructed solution remains feasible. However, since $\lambda_n$ has increased, the constructed solution has a value higher than the optimal value, which is a contradiction.
	
	Suppose $p_k < 1$ and define, for $x \in \mathbb{R}$,
	$$\begin{array}{c}
		\lambda'_n(x) := \lambda_n + x, \  y(x) := - \frac{x}{p_k}\cdot \prod\limits_{j = k}^{n-1}(1 - p_j)p_n \text{ and } \lambda'_k(x):= \lambda_k + y(x),
	\end{array}$$
	and the corresponding versions of $F'_i(x)$ as taking $F_i$ when replacing $\lambda_n \gets \lambda'_n(x)$ and $\lambda_k \gets \lambda_k'(x)$. Additionally, let us define $\lambda_i(x) = \lambda_i$ for $i\in[n-1]\setminus\{k\}$. Notice that,
	\begin{align*}
		F_k'(x) & = p_k \lambda_k'(x) + (1-p_k)F_{k+1}'(x)\\
		& = p_k \lambda_k + p_k y(x) + (1-p_k)F_{k+1} + (1-p_k)\prod\limits_{j=k+1}^{n-1}(1-p_j)p_nx \\
		& = F_k - p_k\cdot \frac{x}{p_k}\cdot \prod\limits_{j = k}^{n-1}(1 - p_j)p_n + \prod\limits_{j=k}^{n-1}(1-p_j)p_nx \\
		& = F_k  = 1.
	\end{align*}
	In particular, we obtain that for any value of $x$, since $F_k'(x) = F_k$, in particular, $F_i'(x) = F_i \leq 1 $, for any $i < k$ as, by induction, $F_{k-1}'$ can be written as function of $F_k'$, $F_{k-2}'$ as function of $F_{k-1}'$, and so on. Finally, notice that, since $F_{k+1} < 1$, $p_k < 1$, and $p_k\lambda_k + (1-p_k)F_{k+1} = F_k = 1$, it holds $\lambda_k > 1$. In particular, since for any $i > k$, $F_i < 1$, we can choose $x > 0$ small enough such that for any $i > k$, $F_i'(x) \leq 1$ and $\lambda_k'(x) \geq 1$, defining therefore a feasible solution of \eqref{eq:B_n'}, whose value is given by
	\begin{align*}
		\sum\limits_{i=1}^n\prod\limits_{j>i}(1-p_j)p_i\lambda_i'(x) &= \sum\limits_{i=1}^n\prod\limits_{j>i}(1-p_j)p_i\lambda_i + \prod\limits_{j>n}(1-p_j)p_nx + \prod\limits_{j>k}(1-p_j)p_k y(x) \\
		&= \sum\limits_{i=1}^n\prod\limits_{j>i}(1-p_j)p_i\lambda_i + p_nx - \prod\limits_{j>k}(1-p_j)p_k \cdot \frac{x}{p_k}\cdot\prod\limits_{j = k}^{n-1}(1-p_j)p_n \\
		&= \sum\limits_{i=1}^n\prod\limits_{j>i}(1-p_j)p_i\lambda_i + xp_n \left( 1 - \prod_{j>k}(1-p_j) \prod\limits_{j = k}^{n-1}(1-p_j) \right)\\
		&> \sum\limits_{i=1}^n\prod\limits_{j>i}(1-p_j)p_i\lambda_i,
	\end{align*}
	where the last inequality holds as $k < n$, $p_i \geq p > 0$ for any $i \in [n]$, and $x > 0$. Since the constructed solution has higher value, we obtain a contradiction.
\end{proof}

Finally, we prove we can take all $F_i$ to be equal to $1$ and, in particular, for all except the last random variable, unit scale factors $\lambda_i = 1$.

\begin{lemma}\label{lem:1 for all i}
	There exists an optimal solution of problem \eqref{eq:B_n'} such that $F_i = 1$ for any $i\in[n]$. Moreover, $\lambda_i = 1$ for all $i \in [n]$.
\end{lemma}
\begin{proof}
	By backward induction on $i$, we prove that we can take $F_i = 1$. Lemma \ref{lem:first is one} proved the initialization, i.e. we can take $F_n = 1$. Then, for $i\in[n-1]$, we take an optimal solution of \eqref{eq:B_n'} with $F_j = 1$ for $j > i$. On the one hand, as optimal solutions are feasible, $F_i \leq 1 $. On the other hand, as $\lambda_i \geq 1$ and $F_{i+1} = 1$, we also have that $F_i \geq p_i + (1-p_i) = 1$. By induction, it follows that there exists an optimal solution with $F_i = 1$ for $i\in[n-1]$. In such optimal solution, for $i \in [n]$, as $F_i = F_{i+1} = 1$ , we have that 
	$$\lambda_i = \frac{1 - (1-p_i)}{p_i} = 1,$$ 
	by the definition of $F_i$ in Equation (\ref{rel:rec F}).
\end{proof}

\noindent We are ready to prove Theorem \ref{thm:CR_bernoulli}.

\begin{proof}[Proof of Theorem \ref{thm:CR_bernoulli}]
	Based on all the previous results, there exists an optimal solution of \eqref{eq:B_n'} with $\lambda_i = 1$ for $i < n$ and $p_n\lambda_n = 1$. Since this solution is also a feasible solution of \eqref{eq:B_n}, by Lemma \ref{lem:relaxation B_n}, it is also an optimal solution of \eqref{eq:B_n}. Notice that, since the first value is deterministically $1$, both $\classdp$ and $\opt$ and thus the competitive ratio, remain unchanged if we add or remove variables with value $1$ for $i \geq 2$.
	In particular, we can consider an instance with only two variables, whose optimal value is the solution to the problem
	\begin{align*}
		\begin{array}{cc}
			\max & p_2\lambda_2 + (1-p_2) \\
			s.t. &  p_2 \in [p,1]\\
			& \lambda_2 \geq 1\\
			& p_2\lambda_2 \leq 1
		\end{array}
	\end{align*}
	that is, $2-p$. Since $\classdp$ obtains $1$, we conclude the proof of the main claim of the theorem. 
	
	Finally, the fact that the bound can be achieved by using a single-threshold algorithm is a consequence of Lemma \ref{lem:ordered -> threshold} and an alternative version of Lemma \ref{lem:ordered setting} for threshold algorithms. Indeed, consider an instance $\Inst$ as stated in Theorem \ref{thm:CR_bernoulli} such that for some $i \in [n]$, $\lambda_i > \lambda_{i+1}$. Notice only two cases matter: (1) The realization of both random variables $X_i$ and $X_{i+1}$ are positive and above the threshold, in which case swapping $X_i$ and $X_{i+1}$ decreases the payoff of the algorithm; and (2) in any other case, the algorithm obtains the same payoff in the original instance or on the swapped one. Thus, ordered instances are the worst instances for threshold algorithms as for $\classdp$, and therefore the sought result is a consequence of Lemma \ref{lem:ordered -> threshold}.
\end{proof}



\end{document}